\documentclass[conference,11pt,onecolumn]{IEEEtran}
\usepackage{graphicx}
\usepackage{amsmath,amssymb,amsthm}    
\usepackage{float}
\usepackage{subfig,graphicx}
\usepackage[linesnumbered,ruled,vlined]{algorithm2e}
\usepackage{comment}
\usepackage{color}
\newtheorem{defn}{Definition}
\newtheorem{thrm}{Theorem}
\newtheorem{cor}{Corollary}
\newtheorem{lma}{Lemma}

\DeclareMathOperator*{\minimize}{\text{minimize}\,}

\DeclareMathOperator*{\maximize}{\text{maximize}\,}
\DeclareMathOperator*{\argmax}{\text{arg max}\,}

\title{Framing Effects on Strategic Information Design under Receiver Distrust and Unknown State}

\author{
\IEEEauthorblockN{Doris E. M. Brown, Venkata Sriram Siddhardh Nadendla}
\\[-1ex]
\IEEEauthorblockA{Computer Science Department \\
Missouri University of Science and Technology\\
Rolla, Missouri \\
Email: \{deby3f, nadendla\}@mst.edu}
}

\begin{document}

\maketitle              
\begin{abstract}
Strategic Information Design is a framework where a sender designs information strategically to steer its receiver's decision towards a desired choice. Traditionally, such frameworks have always assumed that the sender and the receiver comprehends the state of the choice environment, and that the receiver always trusts the sender's signal. This paper deviates from these assumptions and re-investigates strategic information design in the presence of distrustful receiver and when both sender and receiver cannot observe/comprehend the environment state space. Specifically, we assume that both sender and receiver has access to non-identical beliefs about choice rewards (with sender's belief being more accurate), but not the environment state that determines these rewards. Furthermore, given that the receiver does not trust the sender, we also assume that the receiver updates its prior in a non-Bayesian manner. We evaluate the Stackelberg equilibrium and investigate effects of information framing (i.e. send complete signal, or just expected value of the signal) on the equilibrium. Furthermore, we also investigate trust dynamics at the receiver, under the assumption that the receiver minimizes regret in hindsight. Simulation results are presented to illustrate signaling effects and trust dynamics in strategic information design.

\end{abstract}


\section{Introduction}


Strategic information design is a framework where a sender-receiver pair interact with each other with mismatched motives in the presence of informational asymmetries. The sender designs information strategically and sends it to the receiver to steer the receiver's decision in sender's favor. Such a interaction framework has a diverse range of applications ranging from marketing and politics, to even cyber-physical-human systems such as intelligent transportation systems with connected and autonomous vehicles being guided by city's transportation infrastructure. Therefore, this framework has been investigated in diverse domains using several labels such as strategic information transmission \cite{Crawford1982}, cheap talk \cite{Farrell1996} and Bayesian persuasion \cite{Kamenica2011} in economics, optimal information disclosure in politics \cite{Rayo2010}, and strategic information design (along with all the other labels mentioned) in computer science \cite{Dughmi2017,Das2017,Oudah2018,Babichenko2021,Nachbar2020}. 

Due to the wide applicability of strategic information design framework, there is also significant amount of literature available on this topic, particularly in the economics domain. However, due to rich interdisciplinary collaborations across the domains of economics and computer science, strategic information design has also been investigated by computer scientists from a computing viewpoint. Some of the initial attempts include the work by Dughmi (refer to \cite{Dughmi2017} and citations within), where he and his team has investigated the computational complexity of designing information strategically under different interaction settings. Later,  More recently, strategic information transmission has been used in the design of intelligent transportation systems in \cite{Das2017} where Das \emph{et al.} have developed strategic information signals to steer individual travelers' routing decisions with the goal of reducing traffic congestion. More recently, there has been a sudden increase in attention to strategic information design in the context of strategic classification \cite{Sundaram2021}, audit games \cite{Yan2020}, bandit settings \cite{Feng2020}, dynamic interaction with regret-minimizing agents \cite{Babichenko2021} and its implications to price of anarchy \cite{Nachbar2020} as well. However, almost all frameworks make strong assumptions about the interacting agents. For example, both sender and receiver are assumed to have complete knowledge about the choice environment, i.e. its state space and how signaling can be designed based on state information. However, in some complex choice environments (e.g. human-robot interaction during navigation on a large and dynamic graph, where neither human, nor robot has complete knowledge about the graph), both sender and receiver may not have complete information about the state space of the choice environment. 

Furthermore, most frameworks assume that the receiver updates its belief using Bayes rule, which means that the receiver completely trusts the sender's information. However, misalignment in motives between sender and receiver can lead to receiver distrust, thereby decreasing the ability of the system to effectively persuade its users. Trust has been extensively studied in the context of traditional recommender systems literature \cite{odonovan2005,barnett2005}. In fact, several mitigation techniques based on interpreting and/or explaining recommendations have been proposed in the past \cite{Dong2010,Zhang2020,Tintarev2007,Zhang2020,Herlocker2000,ferwerda2018}. However, there is very little attention given to receiver trust and its distorting effects on strategic information design at the sender. It is natural to expect that distrust leads to deviations from Bayesian updating, which is investigated in this paper by constructing receiver's posterior beliefs based on convex combination of prior and signalled beliefs. As evident in real-world situations \cite{griffin1992}, a human's measure of trust can lead to overvaluing or undervaluing new information when constructing an updated belief. As noted in \cite{epstein2006}, while Bayesian updating is deemed to be a standard model, updating behavior can be understood differently when trust probabilities are treated as subjective entities by a human receiver, who has the autonomy to either completely/partially accept, or totally reject the sender's signal depending on their trust levels \cite{deryugina2013}. However, one interesting paper to note with respect to deviations to Bayesian updating (not in the context of receiver distrust) is by Guo \emph{et al.} in \cite{guo2018}, which seeks to resolve miscalibration between a sender's prediction and the true realized distribution of state rewards using a novel belief update rule. While this closely relates to the work done in this paper, \cite{guo2018} continues to rely on Bayesian updating, making it inappropriate when receiver exhibits distrust regarding the sender's signals. 
This paper addresses some practical constraints in strategic information design, particularly regarding limitations at both sender and receiver, which have not been addressed \emph{a priori} in the literature. We model strategic interaction between the sender and the receiver as a Stackelberg signaling game, where both agents have access to non-identical (prior) belief distributions regarding choice rewards, but does not have access to environment state. We assume that the sender constructs a reward belief signal to steer the receiver's decision, who updates its posterior belief based on its trust regarding the sender, its own prior reward belief and the sender's signal without the use of Bayes rule. We compute the equilibrium strategies at both the sender and the receiver, and investigate conditions under which (i) the sender reveals manipulated information, and (ii) receiver trust deteriorates when the true rewards are realized in hindsight.


\section{Model}
Consider a strategic information design setting with a choice environment with a set of choices $\mathcal{N} = \{1, \cdots, N\}$, as shown in Figure 1. Let $\boldsymbol{x} = \{x_1, ..., x_N\} \in \mathcal{X}$ denote the reward profile corresponding to the choices in $\mathcal{N}$. We assume that both the sender (Alice) and receiver (Bob) do not have access to choice state space, but have private and incomplete information regarding the choice rewards in the form of prior beliefs $p(\boldsymbol{x})$ and $q(\boldsymbol{x})$ respectively. In order to make this interaction sensible, we assume that Alice can have access to extrinsic private information which is typically acquired through some sensing infrastructure (e.g. sensor network, social sensing), in order to compute her posterior belief $p(\boldsymbol{x})$. This introduces information asymmetry in our problem setting, which motivates Bob to rely on Alice's messages.

\begin{figure*}[!t]
\centering
\includegraphics[width=1\textwidth]{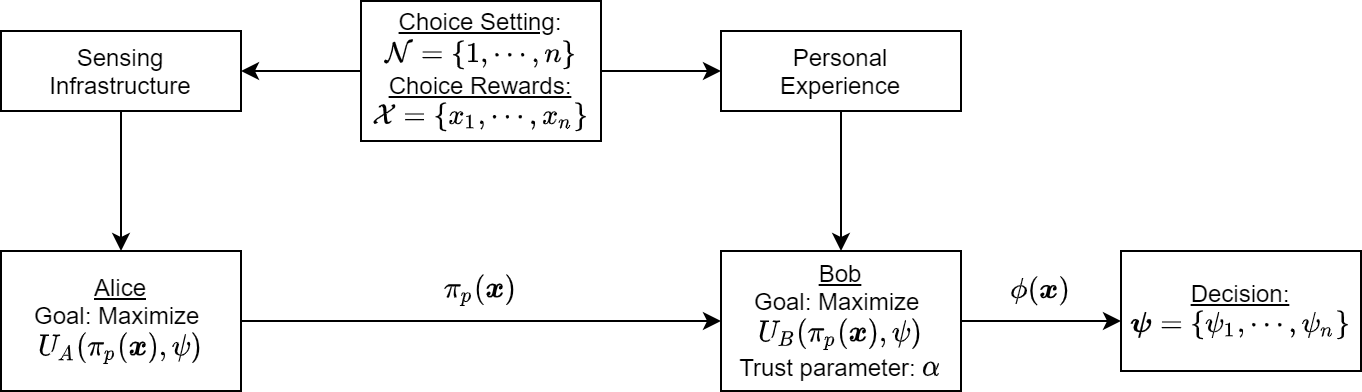}
\caption{System Model and Interactions}
\label{fig:my_label}
\end{figure*}

Assume that both Alice and Bob are expected utility maximizers. In an attempt to maximize her own utility, Alice constructs a new belief signal $\pi_p(\boldsymbol{x})$ over the simplex $\mathcal{S}_{\mathcal{X}}$ based on Alice's belief $p(\boldsymbol{x})$ and shares it with Bob. Then, Bob combines his prior belief with the received information and constructs a posterior belief
\begin{equation}
\phi(\boldsymbol{x}) = \alpha \pi_p(\boldsymbol{x}) + (1 - \alpha) q(\boldsymbol{x}),
\label{Eqn: Posterior Belief}
\end{equation}
where $\alpha \in $ [0, 1] is a parameter that captures Bob's trust in Alice's message. For example, if $\alpha \rightarrow 1$, then Bob starts trusting Alice blindly via disregarding his own prior belief regarding choice rewards. On the other extreme, if $\alpha \rightarrow 0$, then Bob starts distrusting Alice and makes decisions that are totally based on his own prior belief.

Let Bob's state be denoted by the tuple $(\alpha, q(\boldsymbol{x}))$. Let $\boldsymbol{\psi} = \{\psi_1, ..., \psi_N\} \in \mathcal{S}_{\mathcal{N}}$ denote the probabilistic decision rule employed by Bob, where $\mathcal{S}_{\mathcal{N}}$ is the probability simplex on the choice set $\mathcal{N}$, and $\boldsymbol{\psi}_n$ is the probability of picking the $n^{th}$ choice based on Bob's posterior belief $\phi(\boldsymbol{x})$. In such a case, Alice realizes an ex-post
\begin{equation}
U_A(\pi_p(\boldsymbol{x}), \boldsymbol{\psi}) = \displaystyle \sum_{n = 1}^{N} \psi_n \cdot \mathbb{E}_{p}(x_n),
\label{Eqn: Alice Ex-Post Utility}
\end{equation}c
where $p(\boldsymbol{x}_n)$ denotes the $n^{th}$ marginal reward distribution in $p(\boldsymbol{x})$ and $\mathbb{E}_{p}(x_n)$ is the $n^{th}$ marginal expectation of $p(\boldsymbol{x})$. On the other hand, Bob's ex-post utility is given by
\begin{equation}
U_B(\pi_p(\boldsymbol{x}), \psi) = \sum_{n = 1}^{N} \psi_n \cdot \mathbb{E}_{\phi}(x_n),
\end{equation}
where $\phi(\boldsymbol{x}_n)$ denotes the $n^{th}$ marginal reward distribution in $\phi(\boldsymbol{x})$, and $\mathbb{E}_{\phi}(x_n)$ is the $n^{th}$ marginal expectation of $\phi(\boldsymbol{x})$. For example, in underground mines, miners and the robot experience the same rewards regarding different escape-route choices, when miners attempt to escape from a mine. Due to their limited sensing capabilities, both agents can only observe the state within their neighborhood and cannot 

In this paper, we model the strategic interaction between Alice and Bob as a Stackelberg game with Alice as the leader and Bob as the follower, as shown below:
\begin{equation}
\begin{array}{rcl}
\boldsymbol{\psi}^*(\boldsymbol{\pi_p(\boldsymbol{x})}) & \triangleq & \displaystyle \argmax_{\boldsymbol{\psi}} U_B(\pi_p(\boldsymbol{x}), \boldsymbol{\psi}), \textrm{ and } 
\\[3ex]
\pi^*_p(\boldsymbol{x}) & \triangleq & \displaystyle \argmax_{\pi_p(\boldsymbol{x})} U_A(\pi_p(\boldsymbol{x}), \boldsymbol{\psi}^*(\pi_p(\boldsymbol{x}))).
\end{array}
\tag{P1}
\label{Prob: Stackelberg Game}
\end{equation}

\section{Equilibrium Analysis}
In this section, we employ backward induction to solve the game defined in Problem \ref{Prob: Stackelberg Game}. Since Alice is the leader and Bob is the follower, we first evaluate Bob's best response, followed by Alice's optimal signaling strategy using Bob's best response.

\subsection{Bob's Best Response}
Given that Alice chooses a signaling strategy $\pi_p(\boldsymbol{x})$, Bob's best response is to choose $\boldsymbol{\psi} = \{ \psi_1, \cdots, \psi_N \}$ such that the expected utility at Bob
\begin{equation}
U_B(\pi_p(\boldsymbol{x}), \boldsymbol{\psi}) \ = \  \displaystyle \sum_{n = 1}^N \psi_n \displaystyle \mathbb{E}_{\phi}(x_n) \ = \  \displaystyle \sum_{n = 1}^N \psi_n \displaystyle \left[ \alpha \mathbb{E}_{\pi_p}(x_n) + (1-\alpha) \mathbb{E}_{q}(x_n) \right]
\label{Eqn: Bob Expected Utility}
\end{equation}
is maximized.

For the sake of easy notation, let us denote
\begin{equation}
y_n = \displaystyle \alpha \mathbb{E}_{\pi_p}(x_n) + (1-\alpha) \mathbb{E}_{q}(x_n),
\label{Eqn: y_n}
\end{equation}
for all $n \in \mathcal{N}$. Then, the expected utility at Bob can be rewritten as
\begin{equation}
U_B(\pi_p(\boldsymbol{x}), \boldsymbol{\psi}) = \displaystyle \sum_{n = 1}^N \psi_n y_n.
\end{equation}
In other words, the first optimization problem in \eqref{Prob: Stackelberg Game} (which Bob is interested to solve) reduces to
\begin{equation}
\begin{array}{rl}
\displaystyle \minimize_{\boldsymbol{\psi}} & \displaystyle - \boldsymbol{y}^T \boldsymbol{\psi}
\\[3ex]
\text{subject to} & \text{1. } \displaystyle \boldsymbol{1}^T \boldsymbol{\psi} = 1,
\\[2ex]
& \text{2. } \displaystyle \boldsymbol{\psi} \geq \boldsymbol{0},
\end{array}
\tag{P2}
\label{Prob: Bob's Best Response}
\end{equation}
where $\boldsymbol{y} = \{ y_1, \cdots, y_N \}$ is a vector of $y_n$ variables defined in Equation \eqref{Eqn: y_n}.

\begin{lma}
For a given trust parameter $\alpha$, signaling strategy $\pi_p(\boldsymbol{x})$ and prior $q(\boldsymbol{x})$, Bob's best response (i.e. solution to Problem \ref{Prob: Bob's Best Response}) is given by
\begin{equation}
\psi_{n^*}(\pi_p(\boldsymbol{x})) = 
\begin{cases}
1, & \text{if } n = \displaystyle \argmax_{n \in \mathcal{N}} y_n,
\\[2ex]
0, & \text{otherwise},
\end{cases}
\label{Eqn: Bob - Best Response}
\end{equation}
where $y_n = \displaystyle \alpha \ \mathbb{E}_{\boldsymbol{\pi_p}}(x_n) + (1-\alpha) \ \mathbb{E}_{q}(x_n)$.
\label{Thrm: Bob Best Response}
\end{lma}


\begin{proof}
 The Lagrangian function for Problem \eqref{Prob: Bob's Best Response} is given by
 \begin{equation}
 \begin{array}{lcl}
 L(\boldsymbol{\psi}, \boldsymbol{\lambda}, \nu) & = & - \boldsymbol{y}^T \boldsymbol{\psi}  + \nu \left( \boldsymbol{1}^T \boldsymbol{\psi} - 1 \right) - \boldsymbol{\lambda}^T \boldsymbol{\psi}
 \\[2ex]
 & = & \left[ - \boldsymbol{y}  + \nu \boldsymbol{1} - \boldsymbol{\lambda} \right]^T \boldsymbol{\psi} - \nu
 \end{array}
 \label{Eqn: Bob-Lagrangian}
 \end{equation}

 The dual function for Problem \eqref{Prob: Bob's Best Response} is given by
 \begin{equation}
 \begin{array}{lcl}
 \ell(\boldsymbol{\lambda}, \nu) & = & \displaystyle \minimize_{\boldsymbol{\psi} \in \mathcal{S}_{\mathcal{N}}} \ L(\boldsymbol{\psi}, \boldsymbol{\lambda}, \nu)
 \\[3ex]
 & = & 
 \begin{cases}
 - \nu, & \text{if } - \boldsymbol{y} + \nu \boldsymbol{1} - \boldsymbol{\lambda} \geq 0,
 \\[1ex]
  -\infty, & \text{otherwise}.
 \end{cases}
 \end{array}
 \label{Eqn: Bob-Dual}
 \end{equation}

 Note that, for all $\boldsymbol{\lambda} \succeq 0$, the above dual function acts as a lower bound to the Lagrangian function in Equation \eqref{Eqn: Bob-Lagrangian}, which itself acts as a lower bound to the objective function $- \boldsymbol{y}^T \boldsymbol{\psi}$.

 Therefore, the dual problem to Problem \ref{Prob: Bob's Best Response} is given as follows:
 \begin{equation}
 \begin{array}{rl}
 \displaystyle \maximize_{\boldsymbol{\lambda}, \ \nu} & \displaystyle - \nu
 \\[3ex]
 \text{subject to} & \text{1. } \displaystyle \nu \geq y_n + \lambda_n, \text{ for all } n \in \mathcal{N},
 \\[2ex]
 & \text{2. } \displaystyle \lambda_n \geq 0, \text{ for all } n \in \mathcal{N}.
 \end{array}
 \tag{P3}
 \label{Prob: Bob's Best Response - Dual}
 \end{equation}

 Without any loss of generality, Constraint 1 in Problem \eqref{Prob: Bob's Best Response - Dual} can be equivalently replaced with the statement
 \begin{equation}
 \nu \ \geq \ \displaystyle \max_{n \in \mathcal{N}} \  \left( y_n + \lambda_n \right).
 \end{equation}
 Since the objective of Problem \eqref{Prob: Bob's Best Response - Dual} is equivalent to minimizing $\nu$, the optimal choice of $\nu$ reduces to
 \begin{equation}
 \nu^* = \displaystyle \max_{n \in \mathcal{N}} \ y_n.
 \end{equation}
 Since the duality gap in a linear program is zero, the optimal value of the primal problem in \eqref{Prob: Bob's Best Response} is also equal to $\nu^*$, which can be obtained with Bob's best response shown in Equation \eqref{Eqn: Bob - Best Response}.
\end{proof}

In other words, Bob's optimal choice that maximizes his expected utility according to his posterior belief, is a singleton if there is a unique maximum. However, if there are multiple optimal choices, then Bob can randomize to choose any of the optimal choices as they all produce the same outcome.

\subsection{Alice's Optimal Strategy}
Alice's optimal signaling strategy is to send a signal $\pi_p(\boldsymbol{x})$ such that 
\begin{equation}
U_A(\pi_p(\boldsymbol{x}), \boldsymbol{\psi}) = \displaystyle \sum_{n = 1}^{N} \psi_n \cdot \mathbb{E}_{p}(x_n)
\label{Eqn: Alice Optimal Strategy}
\end{equation}
is maximized. Due to the specific structure of Alice's expected utility, we could not find a closed-form expression for Alice's signaling strategy. However, the following theorem provides a necessary condition for optimal strategy at the sender.

\begin{thrm}
The optimal signaling strategy for Alice is to choose any distribution $\pi_p(\boldsymbol{x})$ that satisfies 
\begin{equation}
\alpha \mathbb{E}_{\pi_p}(x_{n^*}) + (1 - \alpha) \mathbb{E}_q(x_{n^*}) \ \geq \ \alpha \mathbb{E}_{\pi_p}(x_{n}) + (1 - \alpha) \mathbb{E}_q(x_n)
\label{Eqn: Optimal Signaling Condition}
\end{equation}
for all $n \in \mathcal{N}$, where $n^* = \argmax \mathbb{E}_{p}(\boldsymbol{x})$.
\label{Thrm: Alice Optimal Signaling}
\end{thrm}

\begin{proof}
From Equation \eqref{Eqn: Alice Optimal Strategy}, Alice's optimal signaling strategy is to send a signal $\pi_p(\boldsymbol{x})$ such that $\displaystyle \sum_{n = 1}^{N} \psi_n \cdot \mathbb{E}_{p}(x_n)$ is maximized. Let $n^* = \argmax \mathbb{E}_{p}(\boldsymbol{x})$. By substituting $n^*$ in Equation \eqref{Eqn: Alice Optimal Strategy}, Alice's optimal signaling strategy is to send a signal $\pi_p(\boldsymbol{x})$ such that her expected utility is 
\begin{equation}
\psi_{n^*} \cdot \mathbb{E}_{p}(x_{n^*}).
\end{equation}
Therefore, from Equation \eqref{Eqn: Bob - Best Response}, Alice seeks a signal $\pi_p(\boldsymbol{x})$ such that
\begin{equation}
\psi_{n^*}(\pi_p(\boldsymbol{x})) = 
\begin{cases}
1, & \text{if } n = \displaystyle \argmax_{n \in \mathcal{N}} \mathbb{E}_p(\boldsymbol{x}_n),
\\[2ex]
0, & \text{otherwise}.
\end{cases}
\end{equation}
By substituting such a $\psi_{n^*}$ into Equation \eqref{Eqn: Bob Expected Utility} and simplifying, Bob's expected utility is
\begin{equation}
\alpha \mathbb{E}_{\pi_p}(x_{n^*}) + (1-\alpha) \mathbb{E}_{q}(x_{n^*}),
\end{equation}
where $n^* = \argmax \mathbb{E}_{p}(\boldsymbol{x})$. By expanding the definition of $n^*$, we obtain the result stated in Theorem \ref{Thrm: Alice Optimal Signaling}.
\end{proof}

In other words, Alice will construct a signal $\pi_p(\boldsymbol{x})$ such that Bob's expected utility is maximized according to his posterior belief $\mathbb{E}_{\phi}(\boldsymbol{x})$ by picking the $n^{*-\text{th}}$ choice, where $n^*$ is the optimal choice at Alice. If Alice sends a signal $\pi_p(\boldsymbol{x})$ such that the inequality in Theorem $\ref{Thrm: Alice Optimal Signaling}$ holds, Alice will successfully persuade Bob to adopt the choice that is optimal for her by exploiting Bob's trust.


If Alice has complete knowledge of $\alpha$ and $q(\boldsymbol{x})$, she wishes to construct a signal $\pi_p(\boldsymbol{x})$ that is persuasive to Bob. In other words, Alice wishes to influence Bob to drive his posterior belief $\phi(\boldsymbol{x})$ to $p(\boldsymbol{x})$. We investigate two sufficient conditions under which Alice can persuade Bob and maximize her utility.

\subsubsection{Partial Information Frame:} 

Alice can steer Bob's decisions according to her desire by modifying Bob's utility $U_B$ to become identical to Alice's utility $U_A$. Formally, the sufficient condition to modify Bob's utility $U_B$ into Alice's utility $U_A$ is $\mathbb{E}_\phi(\boldsymbol{x}) = \mathbb{E}_p(\boldsymbol{x})$, as shown below:
\begin{equation}
U_B(\pi_p(\boldsymbol{x}), \boldsymbol{\psi}) \ = \ \displaystyle \sum_{n = 1}^N \psi_n \displaystyle \mathbb{E}_{\phi}(x_n) \ = \ \displaystyle 
\sum_{n = 1}^N \psi_n \displaystyle \mathbb{E}_{p}(x_n) \ = \ U_A(\pi_p(\boldsymbol{x}), \boldsymbol{\psi}).
\end{equation}

If Alice wishes to send a signal that is sufficient to steer Bob's choice towards Alice's desired outcome, then Alice's goal is to minimize the total squared loss function defined below.
\begin{equation}
\displaystyle \int_{x \in \mathcal{X}} \Big( \mathbb{E}_\phi(\boldsymbol{x}) - \mathbb{E}_p(\boldsymbol{x}) \Big)^2 dx = \displaystyle \int_{x \in \mathcal{X}} \Big( \alpha \mathbb{E}_{\pi_p}(\boldsymbol{x}) + (1-\alpha) \mathbb{E}_{q}(\boldsymbol{x}) - \mathbb{E}_p(\boldsymbol{x}) \Big)^2 dx
\label{Eqn: Alice's Squared Loss}
\end{equation}
Since the first expectation term is a linear operator of Alice's signal $\pi_p$, and since the loss function is quadratic in expectation, the loss function is therefore a convex function of $\pi_p$. Consequently, we can minimize the loss function in \ref{Eqn: Alice's Squared Loss} using standard convex optimization techniques. In this paper, we use CVXPY package in our simulation results in Section \ref{Simulation Results}.

\subsubsection{Complete Information Frame:}

Note that revealing the average rewards $\mathbb{E}_{\pi_p}(\boldsymbol{x})$ is sufficient to steer Bob towards Alice's desired choices. Therefore, any further modification is truly unnecessary, and can lead to distrust at Bob regarding Alice's signaling strategies. For example, Alice can steer Bob's decisions in a similar manner by steering the entire belief $\phi(\boldsymbol{x})$ to match with her prior $p(\boldsymbol{x})$. Formally, a stronger sufficient condition to modify Bob's utility into Alice's utility is $\phi(\boldsymbol{x}) = p(\boldsymbol{x}),$ as shown below:
\begin{equation}
U_B(\pi_p(\boldsymbol{x}), \boldsymbol{\psi}) \ = \ \displaystyle \sum_{n = 1}^N \psi_n \displaystyle \mathbb{E}_{\phi}(x_n) \ = \ \displaystyle 
\sum_{n = 1}^N \psi_n \displaystyle \mathbb{E}_{p}(x_n) \ = \ U_A(\pi_p(\boldsymbol{x}), \boldsymbol{\psi}).
\end{equation}

We utilize Kullback-Leibler (KL) divergence as a metric for minimizing the difference between $p(\boldsymbol{x})$ and $\phi(\boldsymbol{x})$. The KL divergence between $p(\boldsymbol{x})$ and $\phi(\boldsymbol{x})$ is given by
\begin{equation}
\begin{array}{lcl}
D_{KL}\left(p(\boldsymbol{x}), \phi(\boldsymbol{x})\right) & = & \displaystyle \int_{\boldsymbol{x} \in \mathcal{X}} p(\boldsymbol{x}) \log \left( \frac{p(\boldsymbol{x})}{\phi(\boldsymbol{x})} \right) d\boldsymbol{x}
\\[4ex]
& = & \displaystyle \int_{\boldsymbol{x} \in \mathcal{X}} p(\boldsymbol{x}) \log \left( \frac{p(\boldsymbol{x})}{\alpha \pi_p(\boldsymbol{x}) + (1-\alpha) q(\boldsymbol{x})} \right) d\boldsymbol{x}
\end{array}
\end{equation}

Given that our desire is to choose $\pi_p(\boldsymbol{x})$ such that $\phi(\boldsymbol{x})$ is as close to $p(\boldsymbol{x})$ as possible, our goal is to
\begin{equation}
\begin{array}{ll}
\displaystyle \minimize_{\pi_p(\boldsymbol{x})} & D_{KL}\left( p(\boldsymbol{x}),\phi(\boldsymbol{x}) \right)
\\[2ex]
\text{subject to} & \ \text{1. } \pi_p(\boldsymbol{x}) \geq 0, \text{ for all } \boldsymbol{x} \in \mathcal{X},
\\[2ex]
& \ \text{2. } \displaystyle \int_{\boldsymbol{x} \in \mathcal{X}} \pi_p(\boldsymbol{x}) d\boldsymbol{x} = 1.
\tag{P4}
 \label{Eqn: Alice-KL Divergence}
\end{array}
\end{equation}

Similar to the partial information frame, the loss function, i.e. KL-Divergence, is a well-known convex function of $p(\boldsymbol{x})$ and $q(\boldsymbol{x})$ which we seek to minimize. Therefore, we can use standard convex optimization techniques to minimize $D_{KL}$. Specifically, in this paper, we use CVXPY package to minimize $D_{KL}$ in Equation \ref{Eqn: Alice-KL Divergence} in our simulation results presented in Section \ref{Simulation Results}.

\section{Trust Dynamics and Strategic Manipulation}
In this section we first define strategic manipulation and state the conditions under which Alice employs strategic manipulation in her interaction with Bob. We then define the regret that Bob incurs after this interaction as a result of not obtaining his desired outcome, and we present an algorithm by which Bob updates his trust according to his regret.

\begin{defn}
Alice employs \textbf{strategic manipulation} if she chooses $\mathbb{E}_{\pi_p}(\boldsymbol{x}) \neq \mathbb{E}_{\boldsymbol{p}}(\boldsymbol{x})$.
\end{defn}

If we replace $\mathbb{E}_p(\boldsymbol{x}_{n^*})$ in Alice's optimal signaling condition in Equation \eqref{Eqn: Optimal Signaling Condition}, we can find settings in which Alice employs strategic manipulation. We state this condition formally in the following corollary.

\begin{cor}{(to Theorem \ref{Thrm: Alice Optimal Signaling})}
Alice adopts strategic manipulation if there exists at least one $n \in \mathcal{N}$ such that the following condition holds true.
\begin{equation}
\alpha \mathbb{E}_p(x_{n^*}) + (1 - \alpha) \mathbb{E}_q(x_{n^*}) \ < \ \alpha \mathbb{E}_p(x_{n}) + (1 - \alpha) \mathbb{E}_q(x_n),
\label{Eqn: Alice Manipulation Condition}
\end{equation}
where $n^* = \argmax \mathbb{E}_{\pi_p}(\boldsymbol{x})$.
\label{Cor: Alice Strategic Manipulation}
\end{cor}

Note that, when $\alpha = 1$, Condition \eqref{Eqn: Alice Manipulation Condition} does not hold true, since $n^* = \argmax \mathbb{E}_{\pi_p}(\boldsymbol{x})$. In other words, when Bob trusts Alice, Alice naturally has the incentive to reveal truthful information to Bob. We state this result formally in the following corollary.

\begin{cor}{(to Theorem \ref{Thrm: Alice Optimal Signaling})}
If $\alpha = 1$, Alice has no incentive to share deceptive information with Bob. 
\label{Cor: Trust Incentive}
\end{cor}

Note that although Alice may reveal truthful information to Bob, he cannot observe if Alice's signaling strategy is congruent with her prior belief $p(\boldsymbol{x})$. This lack of information regarding Alice's prior belief $p(\boldsymbol{x})$ can lead to distrust by Bob regarding Alice, especially when Bob does not obtain his desired outcomes. 

\begin{cor}{(to Theorem \ref{Thrm: Alice Optimal Signaling})}
If $\alpha = 0$, Alice has no incentive to share deceptive information with Bob. 
\label{Cor: Trust Incentive}
\end{cor}

We note that although any signal sent by Alice will not be persuasive to Bob when $\alpha = 0$ in a one-shot interaction, it may be beneficial to Alice in a repeated game setting to share truthful information with Bob in order to increase his trust.

\subsubsection{Regret and Updated Trust:}

Given that Alice and Bob have mismatched prior beliefs regarding the choice set $\mathcal{N}$, Bob experiences some immediate regret after each interaction with Alice for selecting the choice that corresponds with Alice's shared signal rather than the choice that Bob would have selected had he not interacted with Alice. The regret experienced by Bob once the true rewards of the choices have been revealed is given by
\begin{equation}
\begin{array}{lcl}
R_B & =  & \displaystyle U^*_B - U_B (\pi_p^{*}(\boldsymbol{x}), \boldsymbol{\psi}^{*}),
\label{Eqn: Bob's Regret}
\end{array}
\end{equation}
where $U^*_B$ is the utility that Bob would have received had he not interacted with Alice and had only relied on his prior belief $q(\boldsymbol{x})$ to make his decision. 

Note that $R_B$ is not the regret incurred by Bob for not picking the choice with maximum utility. Instead, $R_B$ is the regret Bob incurs due to relying on Alice's signal rather than working entirely with his own prior belief. Given that $R_B$ is unbounded and Bob may not necessarily know the true reward values of all choices, he may not be able to update his trust according to the degree of his regret. Therefore we assume that Bob adjusts his trust $\alpha$ to a value $\alpha^\prime$ that lies between his previous trust value and the trust value that would have been optimal according to his regret. Formally, we define the trust update heuristic as
\begin{equation}
\begin{array}{rcl}
\alpha^\prime & = & 
\begin{cases}
\ \min \Big\{ 1, \ \alpha - \varepsilon \cdot \text{sign}(R_B) \Big\}, & \text{ if } \  R_B \leq 0,
\\[3ex]
\ \max \Big\{ \alpha - \varepsilon \cdot \text{sign}(R_B), \ 0 \Big\}, & \  \text{ otherwise}.
\end{cases}
\end{array}
\end{equation}
where $\varepsilon$ is some number in $[0, 1]$ in order to ensure $\alpha^\prime \in [0, 1]$. In other words, if Bob incurs a regret $R_B > 0$ in hindsight, Bob would have obtained a higher utility if he had not trusted Alice (i.e., $\alpha = 0$). Since Bob cannot completely disregard Alice's advice forever in the future, he cannot update $\alpha$ to zero and completely distrust Alice. Instead, he decreases his trust by a finite value $\varepsilon \in [0, 1]$ to mitigate Alice's influence in subsequent interactions and minimize Bob's cumulative regret over time. 

The effects of these trust dynamics are further discussed in the context of our simulation results in Section 5. Although a more appropriate approach is to model receiver's trust dynamics in repeated interaction settings, we restrict our attention to this one-shot analysis. Repeated games is out of scope of the current paper, and will be considered in our future work. 

\section{Simulation Results}
\label{Simulation Results}
In this section we illustrate the changes in the utility obtained by both Alice and Bob depending on whether Alice chooses to reveal complete information regarding her belief ($\pi_p(\boldsymbol{x})$) or partial information regarding her belief ($\mathbb{E}_{\pi_p}(\boldsymbol{x})$) in the problem setting in which Alice has full knowledge of Bob’s state. We also analyze the dynamics of Bob’s trust parameter $\alpha$ and regret $R_B$ for each of these possible signals and plot the average KL divergence in Alice's prior and shared belief depending on Bob's trust parameter $\alpha$. We simulate this problem setting by considering a set of $n = |\mathcal{N}|$ randomly generated choices available to Bob, each of which has a reward value that has been chosen by sampling a continuous uniform distribution in the range $[0, 10]$ for tractability. We generate $p(\boldsymbol{x})$ by first calculating a vector of perceived rewards at Alice $\boldsymbol{x}_A$ such that $\displaystyle \boldsymbol{x}_{A, n} = \boldsymbol{x}_n + \lambda_n$ where $\lambda_n$ is chosen by sampling a normal distribution with $\mu = 0$ and $\sigma = 1$ for tractability. $p(\boldsymbol{x}_n)$ then represents the probability that the reward $\boldsymbol{x}_n$ will be favored by Alice through the normalization of $\boldsymbol{x}_A$. We calculate $q(\boldsymbol{x})$ in the same fashion with respect to Bob. The following results have been generated by computing 1000 iterations of each signaling solution (i.e. \eqref{Eqn: Alice's Squared Loss} or \eqref{Eqn: Alice-KL Divergence}) for each independent variable in the following experiments. \eqref{Eqn: Alice's Squared Loss} and \eqref{Eqn: Alice-KL Divergence} are solved using the default ECOS solver in the CVXPY library in Python. 

\begin{figure}[!t]
\centering
\includegraphics[width=1\textwidth]{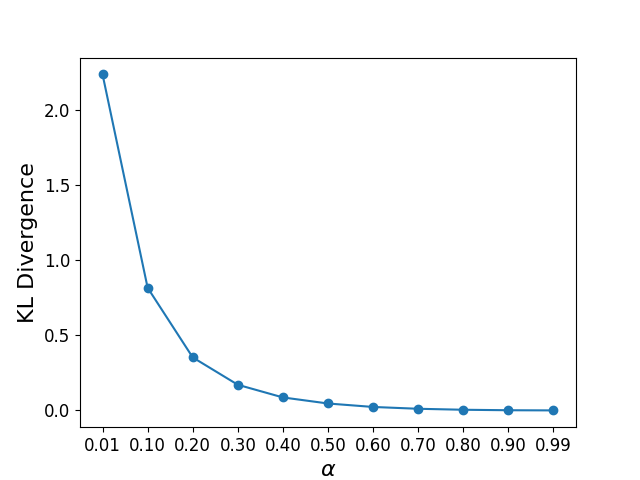}
\caption{Change in the average KL divergence from $\pi_p(\boldsymbol{x})$ to $p(\boldsymbol{x})$ as $\alpha$ increases when calculated across 1000 iterations. For each iteration, the number of choices available to Bob was chosen by sampling a discrete uniform distribution on the interval (2, 20).}\label{Fig:KL Div}
\end{figure}
We use KL divergence as a method of determining the difference in Alice's shared belief $\pi_p(\boldsymbol{x})$ and her prior belief $p(\boldsymbol{x})$ to determine the degree of strategic manipulation that she has employed in her interaction with Bob. From Figure \ref{Fig:KL Div}, we observe the change in the KL divergence from $\pi_p(\boldsymbol{x})$ to $p(\boldsymbol{x})$ as Bob's trust parameter $\alpha$ increases. We note that as $\alpha$ increases, the KL divergence between $\pi_p(\boldsymbol{x})$ and $p(\boldsymbol{x})$ decreases, indicating that as Bob's trust in Alice increases Alice has less incentive to share manipulated information with Bob. Note that $\alpha = 1$ was not considered in Figure \ref{Fig:KL Div} since it follows from equation $\eqref{Eqn: Posterior Belief}$ that when $\alpha = 1$
\begin{equation}
\begin{array}{lcl}
\phi(\boldsymbol{x}) = \alpha \pi_p(\boldsymbol{x}) + (1 - \alpha) q(\boldsymbol{x})
\\[4ex]
\phi(\boldsymbol{x}) = \pi_p(\boldsymbol{x}).
\end{array}
\end{equation}
Therefore, when $\alpha = 1$, $\displaystyle D_{KL}\left(\phi(\boldsymbol{x}), \pi_p(\boldsymbol{x}) \right) = 0$. Given that when $\alpha = 0$ the signal $\pi_p(\boldsymbol{x})$ that Alice sends will have no impact on Bob's decision, we have also chosen to not consider this case in Figure \ref{Fig:KL Div}.

\begin{figure}[!t]
\centering
\subfloat[$\epsilon = 0.1$]{\label{a}\includegraphics[width=0.49\linewidth]{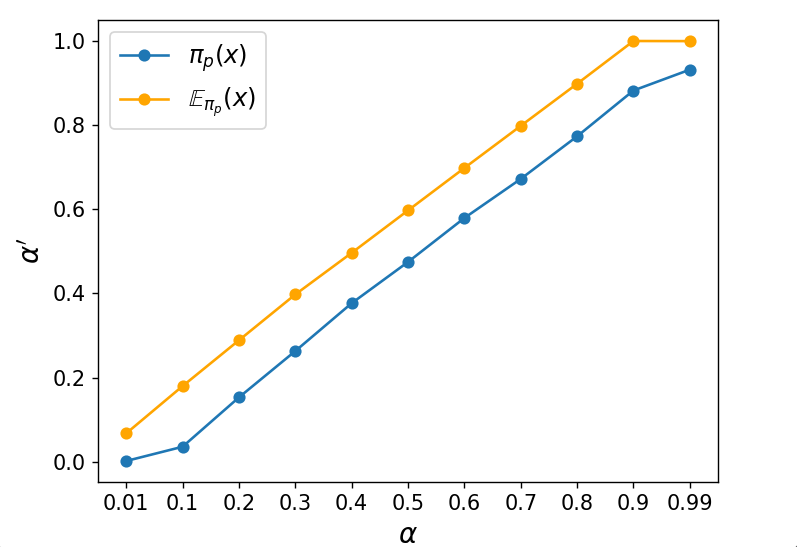}}\hfill
\subfloat[$\epsilon = 0.3$]{\label{b}\includegraphics[width=0.49\linewidth]{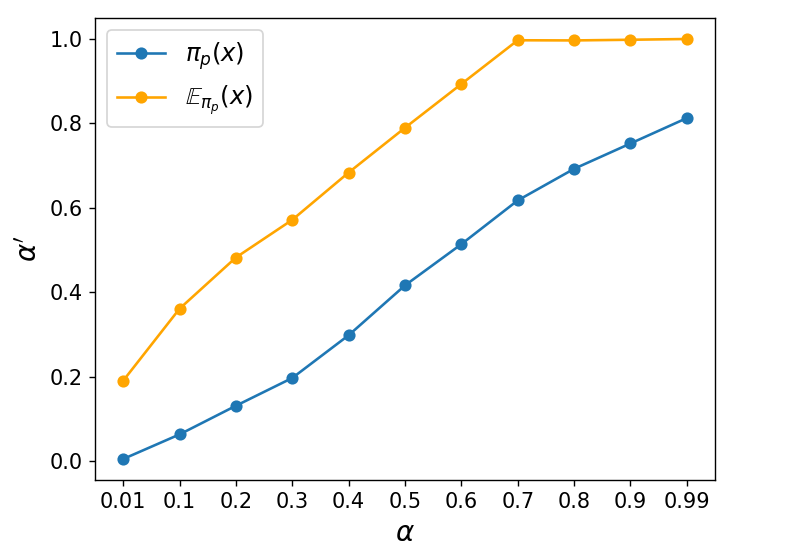}}\par 
\subfloat[$\epsilon = 0.5$]{\label{c}\includegraphics[width=0.49\linewidth]{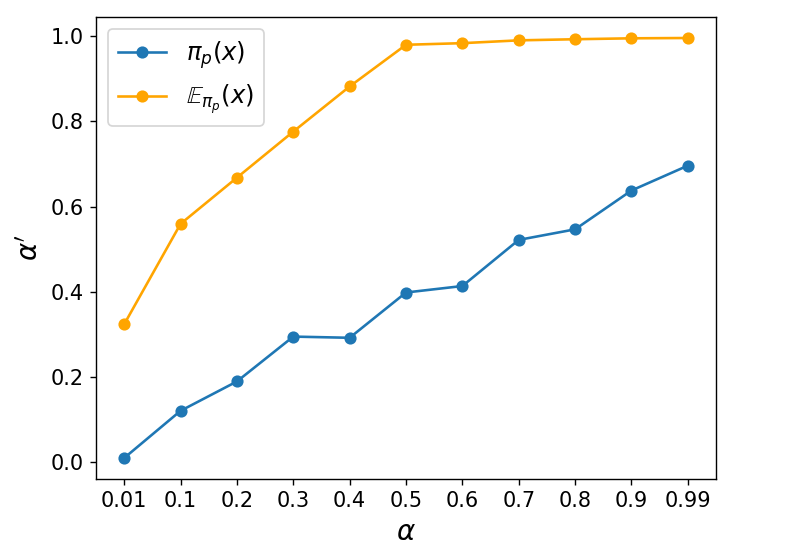}}
\caption{Change from $\alpha$ to the average $\alpha^\prime$ for each possible signal sent by Alice as $\alpha$ increases when calculated across 1000 iterations. For each iteration $n = 10$.}\label{Fig:Trust Update}
\end{figure}
\begin{figure}[!b]
\centering
\vspace{-5ex}
\includegraphics[width=0.7\textwidth]{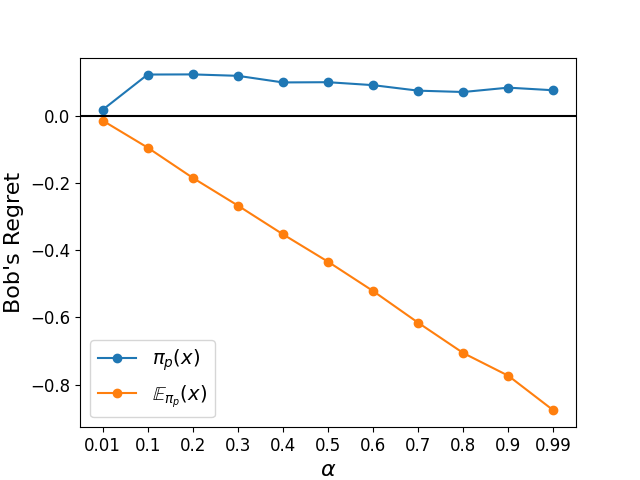}
\caption{Average regret incurred by Bob as $\alpha$ increases for each possible signal sent by Alice when calculated across 1000 iterations. For each iteration $n = 10$.}
\label{fig:Regret}
\end{figure}
\vspace{2cm}
Figure \ref{Fig:Trust Update} illustrates the change from $\alpha$ to $\alpha^\prime$ for each starting value of $\alpha$ at Bob. The results in Figure \ref{Fig:Trust Update} show that the average value of $\alpha^\prime$ is greater when Alice reveals average rewards $\mathbb{E}_{\pi_p}(\boldsymbol{x})$ as opposed to when she reveals the full distribution $\pi_p(\boldsymbol{x})$. Additionally, the difference between Bob's average value of $\alpha^\prime$ when Alice reveals partial information and Bob's average value of $\alpha^\prime$ when Alice reveals complete information increases as the value of $\epsilon$ increases. The results shown in Figure \ref{Fig:Trust Update} imply that Bob experiences a significantly more positive interaction with Alice when she chooses to reveal partial information to Bob. 

From Figure \ref{fig:Regret}, we note that Bob's average regret is always negative and decreases significantly as the value of his trust parameter $\alpha$ increases when Alice reveals average rewards $\mathbb{E}_{\pi_p}(\boldsymbol{x})$. On the other hand, Bob's average regret when Alice reveals the full distribution $\pi_p(\boldsymbol{x})$ remains positive as $\alpha$ increases and decreases only slightly as $\alpha$ approaches 1.

\begin{figure}[!t]
\centering
\includegraphics[width=1\textwidth]{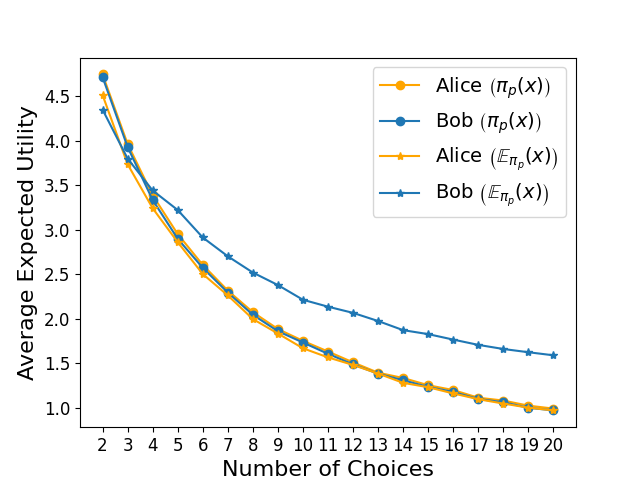}
\caption{Change in average expected utility at Bob as the number of choices available to Bob increases when calculated across 1000 iterations. For each iteration $\alpha = 0.5$.}
\label{fig:Utility}
\end{figure}

Figure \ref{fig:Utility} shows the change in the average expected utility obtained by Alice and Bob for each signaling scheme that Alice can choose as the number of choices available to Bob increase. As the number of choices available ranges from [2, 3], both Bob and Alice receive a higher expected utility when Alice shares the full distribution $\pi_p(\boldsymbol{x})$ with Bob. However, as the number of choices ranges from [4, 20], the expected utility obtained by Alice is roughly equivalent regardless of which signaling scheme she chooses. The expected utility obtained by both Alice and Bob are roughly equivalent throughout this range of choices when Alice reveals the full distribution $\pi_p(\boldsymbol{x})$ to Bob. On the other hand, Bob receives a significantly higher expected utility than Alice when the number of choices available ranges from [4, 20] when Alice reveals expected rewards $\mathbb{E}_{\pi_p}(\boldsymbol{x})$. From these results, we can conclude that when the number of choices available $n > 3$, Alice does not have any incentive to choose one signaling scheme over the other in a one-shot interaction with Bob. However, from Figures \ref{Fig:Trust Update} and \ref{fig:Utility}, Alice may benefit in a repeated interaction setting by choosing to reveal average rewards to Bob thereby increasing the value of his trust parameter $\alpha$.

From the above figures, we note that Bob benefits significantly in this interaction with Alice when Alice chooses to reveal partial information through average rewards $\mathbb{E}_{\pi_p}(\boldsymbol{x})$ rather than a complete information through the full distribution $\pi_p(\boldsymbol{x})$ when the number of choices available to Bob $n > 3$. While revealing a full distribution to Bob provides an explanation to Bob regarding the beliefs of Alice, the significant increases in trust and expected utility that Bob benefits from as a result of Alice choosing to share expected rewards with Bob may be a result of the interaction with Alice lacking transparency, as Bob cannot confirm during the interaction whether Alice's signal $\pi_p(\boldsymbol{x})$ is congruent with her prior belief $p(\boldsymbol{x})$, or a result of cognitive overload at Bob.

\section{Conclusion and Future Work}
In this paper, we investigated the impact of framing effects and trust dynamics on  strategic information design when the state space is unknown to both sender and receiver. We proved that the receiver's best response strategy is a singleton, i.e. the receiver chooses the optimal choice in a deterministic manner, and also show that the sender's optimal strategy is computationally intractable. Therefore, we considered two signaling frames, one where the sender presents partial information, and the other where the sender presents full information to the receiver. We also model trust dynamics based on receiver's regret for accepting sender's information. Numerical results were presented to demonstrate the degree of strategic manipulation employed by the sender depending on receiver's trust, the effects of each signaling frame on receiver's trust, the regret incurred by the receiver, and the expected utility obtained by each agent, all in the case of both signaling frames. In the future, we will extend this work by considering a repeated interaction setting with no-regret dynamics at the sender in which the receiver’s trust parameter $\alpha$ and prior belief  $q(\boldsymbol{x})$ are unknown at the outset of the interaction. We will also develop a better trust dynamics model in the repeated interaction setting where the receiver minimizes cumulative regret, as opposed to instantaneous, one-shot regret in our paper.

%
%
%
 \bibliographystyle{splncs04}
 \bibliography{bibliography}
%
%
%
%
%
\end{document}